\documentclass[copyright,creativecommons]{eptcs}
\providecommand{\event}{GaSCOM 2024} 

\usepackage{iftex}

\ifpdf
  \usepackage{underscore}         
  \usepackage[T1]{fontenc}        
\else
  \usepackage{breakurl}           
\fi

\usepackage{enumitem,stix}
\usepackage{epsf}

\usepackage{hyperref}
\usepackage{mathtools}
\usepackage{ytableau}
\usepackage{todonotes}
\usepackage{verbatim}
\usepackage{amsmath}
\usepackage{amsopn}
\usepackage{amssymb}
\usepackage{amsfonts}
\usepackage{amsthm}
\usepackage{tikz}

\newtheorem{thm}{Theorem}[section]

\newtheorem{obs}[thm]{Observation}
\newtheorem{question}[thm]{Question}

\theoremstyle{definition}

\newtheorem{definition}[thm]{Definition}

\definecolor{fondo}{rgb}{0.898,0.996,0.898}
\definecolor{diagonal}{rgb}{0.466,0,0}
\usetikzlibrary{shapes.geometric}

\title{Interval Posets and Polygon Dissections}
\author{Eli Bagno, Estrella Eisenberg, Shulamit Reches and Moriah Sigron
\institute{Jerusalem College of Technology, 21 HaVaad HaLeumi St., Jerusalem, Israel}
\email{bagnoe@g.jct.ac.il}}
\def\titlerunning{Interval Posets and Polygon Dissections}
\def\authorrunning{Eli Bagno, Estrella Eisenberg, Shulamit Reches and Moriah Sigron}
\def\copyrightholders{E. Bagno, E. Eisenberg, S. Reches and M. Sigron}

\begin{document}

\maketitle

\begin{abstract}
The Interval poset of a permutation is an effective way of capturing all the intervals of the permutation and the inclusions between them and was introduced recently by Tenner. 
Thi paper explores the geometric interpretation of interval posets of permutations. We present a bijection between tree interval posets and convex polygons with non-crossing diagonals, offering a novel geometric perspective on this purely combinatorial concept. 
Additionally, we provide an enumeration of interval posets using this bijection and demonstrate its application to block-wise simple permutations. 
\end{abstract}

\section{Introduction}

In \cite{T}, Tenner  defined the concept of an interval poset of a permutation. 
 This is an effective way of capturing all the intervals of a permutation and the set of inclusions between them in one glance. 
Tenner dealt with structural aspects of the interval poset and characterized the posets $P$ that can be seen as interval posets of some permutations.  
 
 An interval poset might correspond to more than one permutation. For instance, all simple permutations of a given order $n$ share the same interval poset. 
Tener, in the aforementioned paper, enumerated binary interval posets and binary tree interval posets but left open the following question:

\begin{question}
How many tree interval posets have $n$ minimal elements?
\end{question}

This question was answered by Bouvel, Cioni and Izart in \cite{BCL}. They also noted that the number of tree interval posets is equal to the number of ways to place non-crossing diagonals in a convex $(n+1)$-gon such that no quadrilaterals are created. 

In this work we suggest a simple bijection between the set of tree interval posets and the set of $(n+2)$-gons, satisfying the conditions listed above. We use this bijection also for enumerating the whole set of interval posets by using a broader set of polygons. In \cite{BCL}, the enumeration of the entire set of interval posets was done in an algebraic way, using generating functions, while our bijection grants a geometric view to the interval posets. 

Another set of interval posets that can be enumerated by polygons is the one corresponding to block-wise simple permutations, a term that was introduced in a recent paper by the current authors \cite{BERS}.

\section{Background}
\begin{definition}
Let $\mathcal{S}_n$ the symmetric group on $n$ elements. Let $\pi=a_1 \cdots a_n \in \mathcal{S}_n$.   
An {\em interval} (or {\em block}) of $\pi$ is a non-empty contiguous
sequence of entries $a_i a_{i+1} \cdots a_{i+k}$ whose values also form a contiguous sequence of integers.
For $a<b$, $[a,b]$ denotes the interval of values that range from $a$ to $b$. 
Clearly, $[n]:=[1,n]$ is an interval, as well as $\{i\}$ for each $i \in [n]$. These are called trivial intervals. The other intervals are called {\it proper}.
\end{definition}
For example, the permutation $\pi=314297856$ has $[5,9]=97856$ as a proper interval as well as the following proper intervals: $[1,4],[5,6],[7,8],[7,9]$, $[5,8]$.


A permutation $\pi \in \mathcal{S}_n$ is called {\em simple} if it does not have proper intervals. 
For example, the permutation $3517246$ is simple.

Following Tenner \cite{T}, we define an {\it interval poset}
for each permutation as follows:

\begin{definition} \label{def interval poset}
The interval poset of a permutation $\pi \in \mathcal{S}_n$ is the poset $P(\pi)$ whose elements are the non-empty intervals of $\pi$; the order is defined by set inclusion (see for example Figures \ref{fig:IntervalPoset5123647} and \ref{fig:IntervalPoset3142}). 
The minimal elements are the intervals of size $1$.  
\end{definition}

In \cite{T}, the interval poset is embedded in the plane so that each node's direct descendants are increasingly ordered according to the minimum of each interval from left to right. 
We note that in \cite{BCL} another embedding of the same poset was presented.

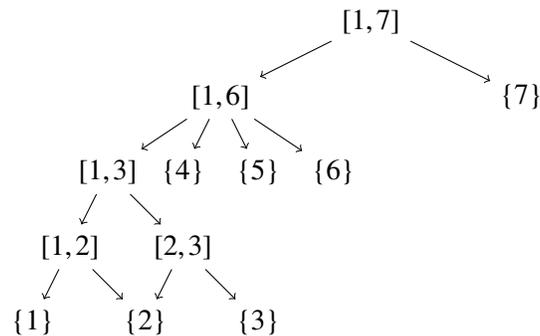
\begin{figure}[!ht]
     \centering

 \begin{tikzpicture}
       \tikzstyle{every node} = [rectangle]
      
         \node (17) at (0,0) {$[1,7]$};
            \node (16) at (-2,-1) {$[1,6]$};
            \node (7) at (2,-1) {$\{7\}$};
            \node (13) at (-3.5,-2) {$[1,3]$};
            \node (12) at (-4,-3) {$[1,2]$};
            \node (23) at (-2.5,-3) {$[2,3]$};
            \node (4) at (-2.5,-2) {$\{4\}$};
            \node (5) at (-1.5,-2) {$\{5\}$};
            \node (6) at (-0.5,-2) {$\{6\}$};
            \node (1) at (-4.5,-4) {$\{1\}$};
            \node (2) at (-3,-4) {$\{2\}$};
            \node (3) at (-1.5,-4) {$\{3\}$};
            
        \foreach \from/\to in {17/16,17/7,16/13,16/4,16/5,16/6,13/12,12/1,23/2,23/3,13/23,12/2}
            \draw[->] (\from) -- (\to);
    \end{tikzpicture}
 
\caption{Interval poset of the permutations: 5123647, 5321647, 4612357, 4632157, 7463215, 7461235, 7532164, 7512364 }
   \label{fig:IntervalPoset5123647}
 \end{figure}

If $\pi$ is a simple permutation, the interval poset of $\pi$ comprises the entire interval $[1,\dots,n]$ with minimal elements $\{1\},\dots,\{n\}$ as its only descendants. Hence, all simple permutations of a given order $n$ share the same interval poset (see for example Figure \ref{fig:IntervalPoset3142}).

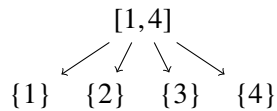
\begin{figure}[!ht]
     \centering

 \begin{tikzpicture}
       \tikzstyle{every node} = [rectangle]
      
         \node (14) at (0,0) {$[1,4]$};
            \node (1) at (-1.5,-1) {$\{1\}$};
            \node (2) at (-0.5,-1) {$\{2\}$};
            \node (3) at (0.5,-1) {$\{3\}$};
            \node (4) at (1.5,-1) {$\{4\}$};

        \foreach \from/\to in {14/1,14/2,14/3,14/4}
            \draw[->] (\from) -- (\to);
    \end{tikzpicture}
 
\caption{Interval poset of permutations 3142 and 2413.  }
     \label{fig:IntervalPoset3142}
 \end{figure} 

\section{Geometrical view of interval posets}
\subsection{General interval posets}

Bouvel, Cioni and Izart\cite{BCL}, provided a formula for the number of interval posets with $n$ minimal elements and added it to OEIS as sequence A348479 \cite{OEIS}.  


Here we provide a geometrical view to the interval posets by providing a bijection from the set of interval posets with $n$ minimal elements to a distinguished set of dissections of the convex $(n+1)-$ gon, which we define below.

We identify a polygon with its set of vertices and denote a diagonal or an outer edge of the polygon 
from vertex $i$ to vertex $j$ by $\{i,j\}$.

\begin{definition}\label{def df}
A dissection of an $(n+1)-$ gon will be called {\it diagonally framed} if for each two crossing diagonals, their vertices are connected to each other. Explicitly, if $\{a,b\}$ and $\{c,d\}$ are two crossing diagonals, 
then the diagonals or outer edges $\{a,d\},\{b,d\},\{c,b\},\{c,a\}$ must also exist.   See Figure \ref{framed} for an example. 
 \end{definition}
Before we proceed, we have to present two observations  
which provide some details on the structure of interval posets and will be used in the sequel. 

\begin{obs}\label{structure of interval posets 1}
Let $\pi \in \mathcal{S}_n$. If $I$ and $J$ are intervals of $\pi$ such that $I \nsubseteq  J$  and $J \nsubseteq I$ and $I\cap J\neq \emptyset$, then $I\cap J$, $I\cup J$, $I-J$ and $J-I$ are intervals of $\pi$.  
\end{obs}

For example, take $\pi=3124576$, then $I=[1,5]$ and $J=[4,7]$ are intersecting intervals of $\pi$ and thus $I\cup J=[1,7],I\cap J=[4,5],I-J=[1,3],J-I=[6,7]$ are also intervals of $\pi$, as can be seen in Figure \ref{blocks} which depicts the permutation $\pi$ in the common graphical way. 

\begin{figure}
    \centering
    \includegraphics[scale=0.45]{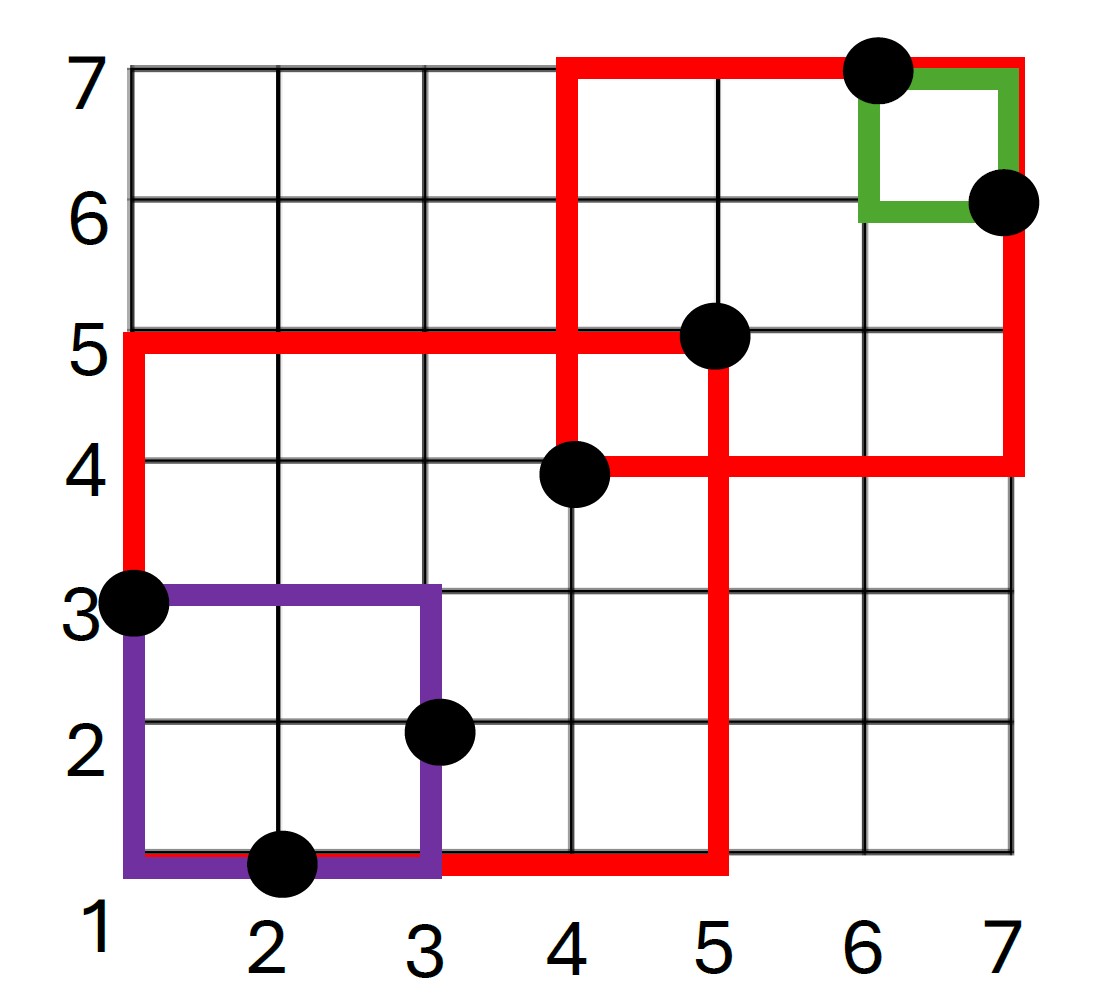}
    \caption{The permutation $\pi=3124576$ and its blocks in a graphical way.}
    \label{blocks}
\end{figure}

\begin{obs}\label{structure of interval posets 2}
If $P(\pi)$ is the interval poset of $\pi\in \mathcal{S}_n$, then no element of $P(\pi)$ has exactly $3$ direct descendants, 
since every permutation of order $3$ must contain a block of order $2$. 
\end{obs}

\begin{figure}[!ht]

    \centering
    \includegraphics[scale=0.35]{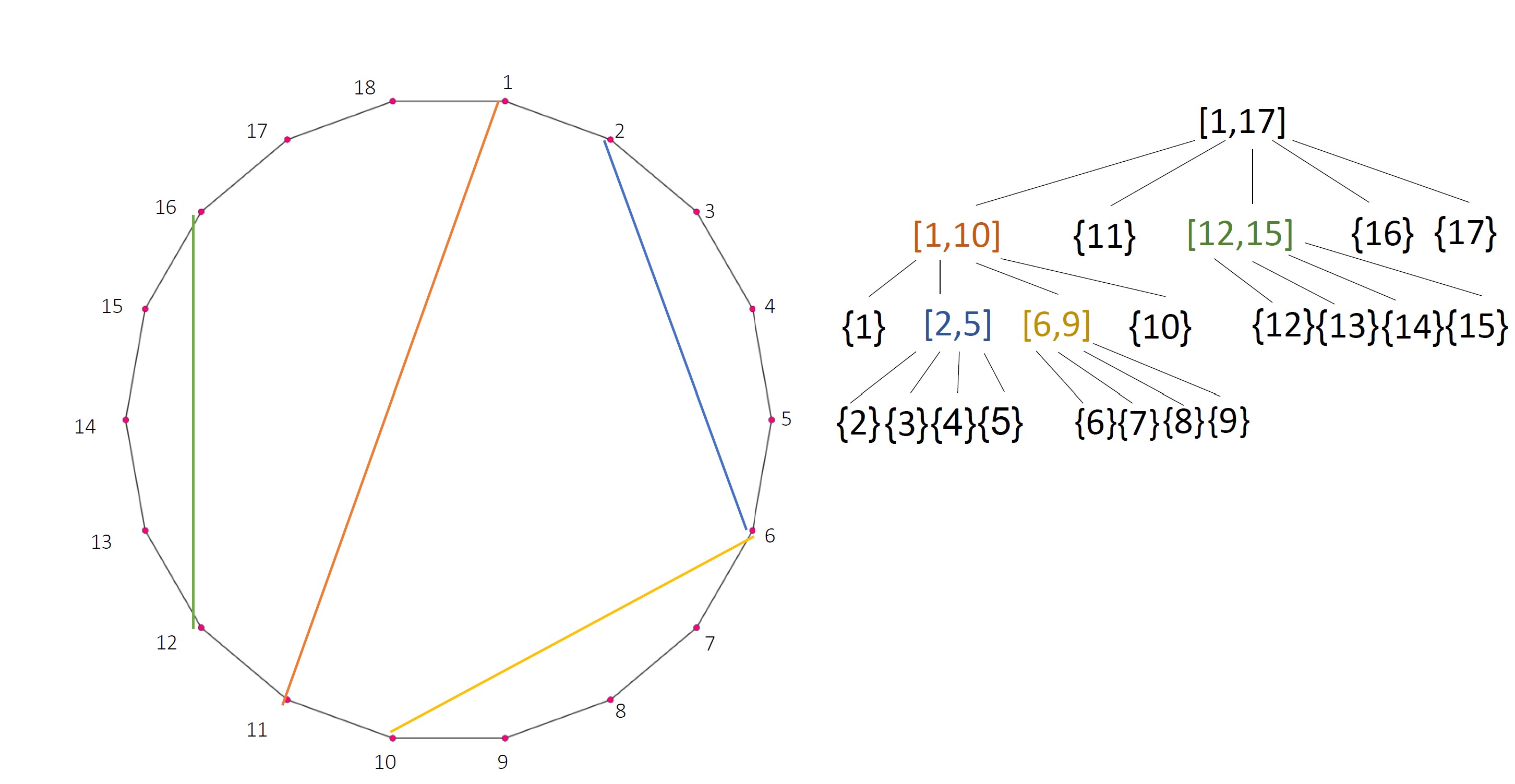}
    \caption{Right: the interval poset P. Left: the polygon $\Phi(P)$}
    \label{polygon to poset}
\end{figure}

We are ready now to present the main result of this subsection.

\begin{thm}
The number of interval posets with $n$ minimal elements is equal to the number of diagonally framed dissections of the convex ($n+1$)-gon such that no quadrilaterals are present (see Figure \ref{all} in the appendix for some examples of the bijection in small values of $n$).
\label{bijection1}
\end{thm}

\begin{proof}
We define a bijection between the set of interval posets with $n$ minimal elements and the set of diagonally framed dissections of convex $(n+1)$- gons without quadrilaterals as follows:

Let $P$ be the interval poset of some $\pi\in \mathcal{S}_n$. We set $\Phi(P)$ to be the convex $(n+1)$-gon whose set of diagonals is $$\{\{a,b+1\}|[a,b] \text{ is an internal node of } P\},$$ i.e. to each interval of the form $[a,b]$ corresponds a diagonal $\{a,b+1\}$ in $\Phi(P)$; note that singletones intervals correspond to outer edges in the polygon  (see Figure \ref{polygon to poset} for an example). 

We claim now that $\Phi(P)$ must be a diagonally framed $(n+1)$ - gon. Indeed, if $\{a,c+1\}$ and $\{b,d+1\}$ are two crossing diagonals in $\Phi(P)$, where $a \leq b \leq c \leq d$, then $I=[a,c]$ and $J=[b,d]$ are intersecting intervals in $P$ and by Observation \ref{structure of interval posets 1} we have that $I\cup J=[a,d]$, $I\cap J=[b,c]$, $I-J=[a,b-1]$ and $J-I=[c+1,d]$ are intervals in $P$ corresponding respectively to the diagonals $\{a,d+1\}, \{b,c+1\},\{a,b\}$ and $\{c+1,d+1\}$. (See Figure \ref{framed} for an illustration). 

Moreover, $\Phi(P)$ must not contain any quadrilateral. Otherwise, if $a<b<c<d$ are such that $\{a.b,c,d\}$ is a quadrilateral (without any subdivision) then $P$ must contain the intervals $[a,b-1],[b,c-1],[c,d-1]$ and $[a,d-1]$ so we must have that the first three intervals are direct descendants of the fourth one and they are the only ones. By Observation \ref{structure of interval posets 2}, this is impossible. 
\end{proof}

\begin{figure}
    \centering
   \includegraphics[scale=0.15]{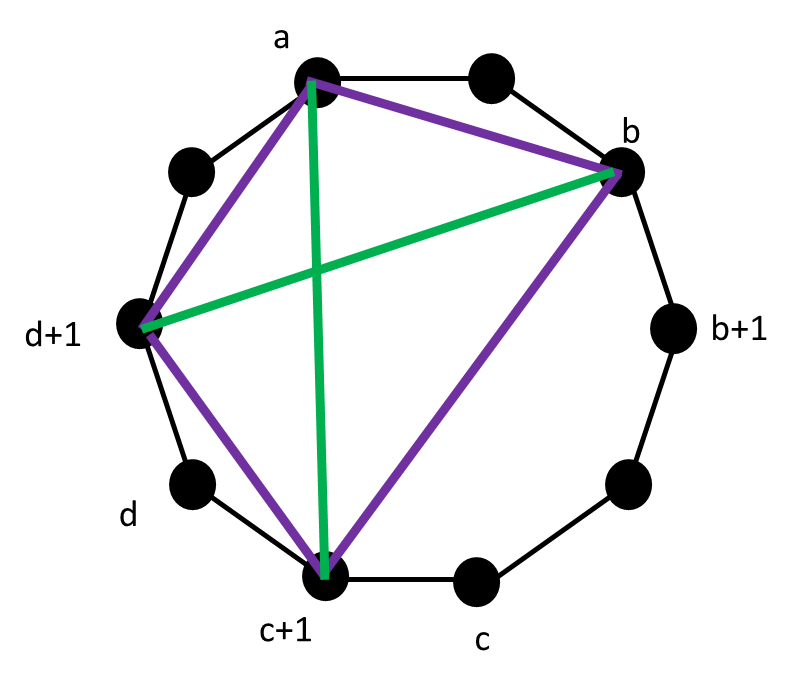}
    \caption{}
    \label{framed}
\end{figure}

\subsection{Tree interval posets}
A {\it tree poset} is a poset whose Hasse diagram is a tree.

In \cite{BCL}, the authors calculated the generating function of the number of tree interval posets using generating functions and mentioned that this is equal to the number of ways to place non-crossing diagonals in a convex $(n+2)$-gon such that no quadrilaterals are created (sequence A054515 from OEIS \cite{OEIS}). 


Using the function $\Phi$ defined above, one can easily produce a combinatorial proof of the following result.  

\begin{thm}\label{num}
The number of tree interval posets with $n$ minimal elements is
equal to the number of non crossing dissections of the convex $(n+1)-$ gon such that no quadrilaterals are present (see Figure \ref{trees} in the appendix for some examples
of the bijection).
\end{thm}

\begin{proof}

We use the same mapping $\Phi$ which was applied in the proof of Theorem \ref{bijection1}. It is now sufficient to prove that no crossing diagonals are obtained. This is implied by the fact that intersecting diagonals stem from intersecting intervals which can not exist in a tree since they cause a circle.  (See Figure \ref{fig:IntervalPoset5123647}). 

\end{proof}
\subsection{Interval posets of block-wise simple permutations}
In \cite{BERS}, the current authors introduced the notion of block-wise simple permutations. We cite here the definition:

\begin{definition}\label{def by 2 blocks}
A permutation $\pi \in \mathcal{S}_n$ is called {\it block-wise simple} if  it has no interval of the form $p_1\oplus p_2$ or $p_1 \ominus p_2$,  where $\oplus$ and $\ominus$ stand for direct and skew sums of permutations respectively. 
\end{definition}

There are no block-wise simple permutations of orders $2$ and $3$.
For $n \in \{4,5,6\}$, a permutation is block-wise simple, if and only if it is simple. One of the first nontrivial examples of block-wise simple permutations is $4253716$.\\  


      
            
 %



In \cite{BERS}, the current authors enumerated the interval posets of block-wise permutations. 


The first few values of the sequence of these numbers  are $1,1,1,5,10,16,45,109,222,540$. This is sequence A054514 from OEIS \cite{OEIS} which also counts the number of ways to place non-crossing diagonals in a convex $(n+4)$-gon such that there are no triangles or quadrilaterals.

The geometrical interpretation of interval posets of block-wise permutations is as follows:

\begin{thm}
The number of interval posets that represent a block-wise simple permutation of order $n$ is equal to the number of ways to place non-crossing diagonals in a convex ($n+1$)-gon such that no triangles or quadrilaterals are present (see Figure \ref{blockwise} in the appendix for some examples
of the bijection).
\end{thm}

\begin{proof}
We use again the mapping $\Phi$, defined earlier. 
In \cite{T} (Theorem 6.1), the author claimed that $P(\sigma)$ is a tree interval poset if and only if $\sigma$ contains no interval of the form $p_1\oplus p_2\oplus p_3$ or $p_1\ominus p_2 \ominus p_3$. 
From here, and by Definition \ref{def by 2 blocks}, it is obvious that an interval poset of a block-wise simple permutation is a tree. Hence it is sufficient to prove that for an interval poset $P$ of a block-wise permutation, $\Phi(P)$ has no triangles. This holds due to the fact that if $\Phi(P)$ contains a triangle with edges $\{a,b\},\{b,c\}, \{a,c\}$ with $a<b<c$ then $P$ must contain the intervals $[a,b-1],[b,c-1]$ and $[a,c-1]$ and thus $[a,c-1]$ is the direct parent of $[a,b-1]$ and $[b,c-1]$ which contradicts the definition of block-wise simple permutations. 
\end{proof}

\bibliographystyle{eptcs}
\bibliography{bagno-bib}

\newpage

\section{Appendix}

\begin{figure}[!ht]
    \centering
   \includegraphics[scale=0.19]{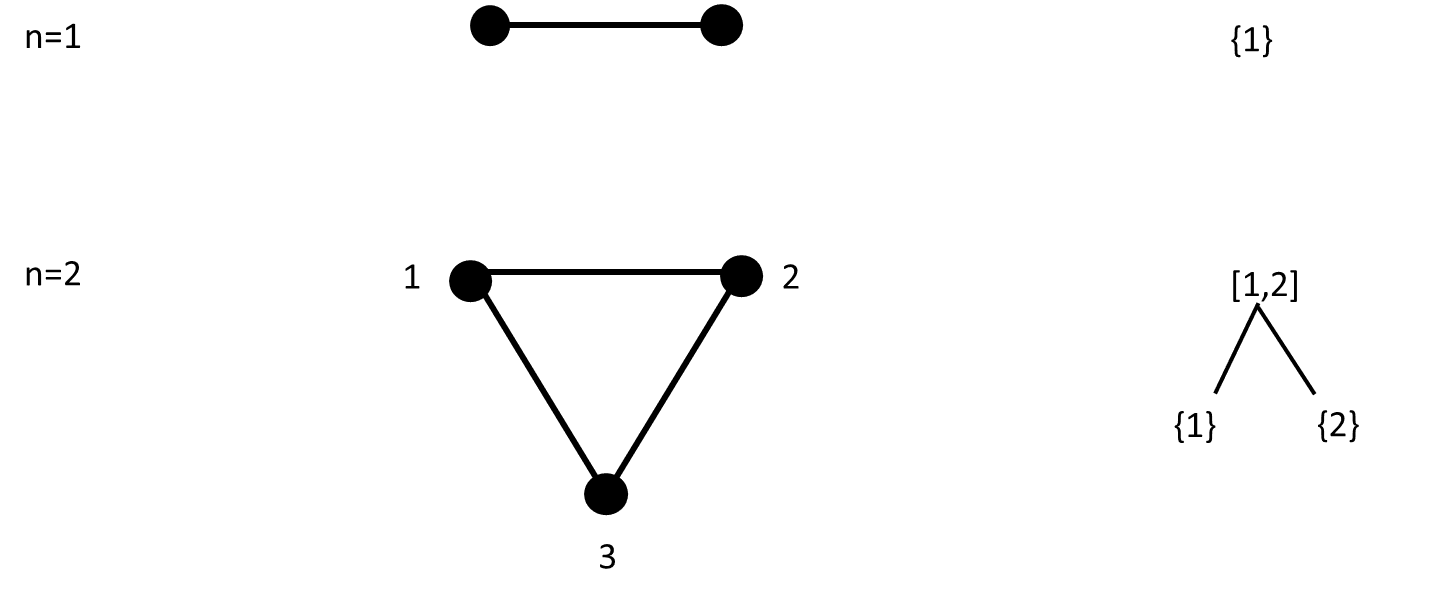}
   \includegraphics[scale=0.19]{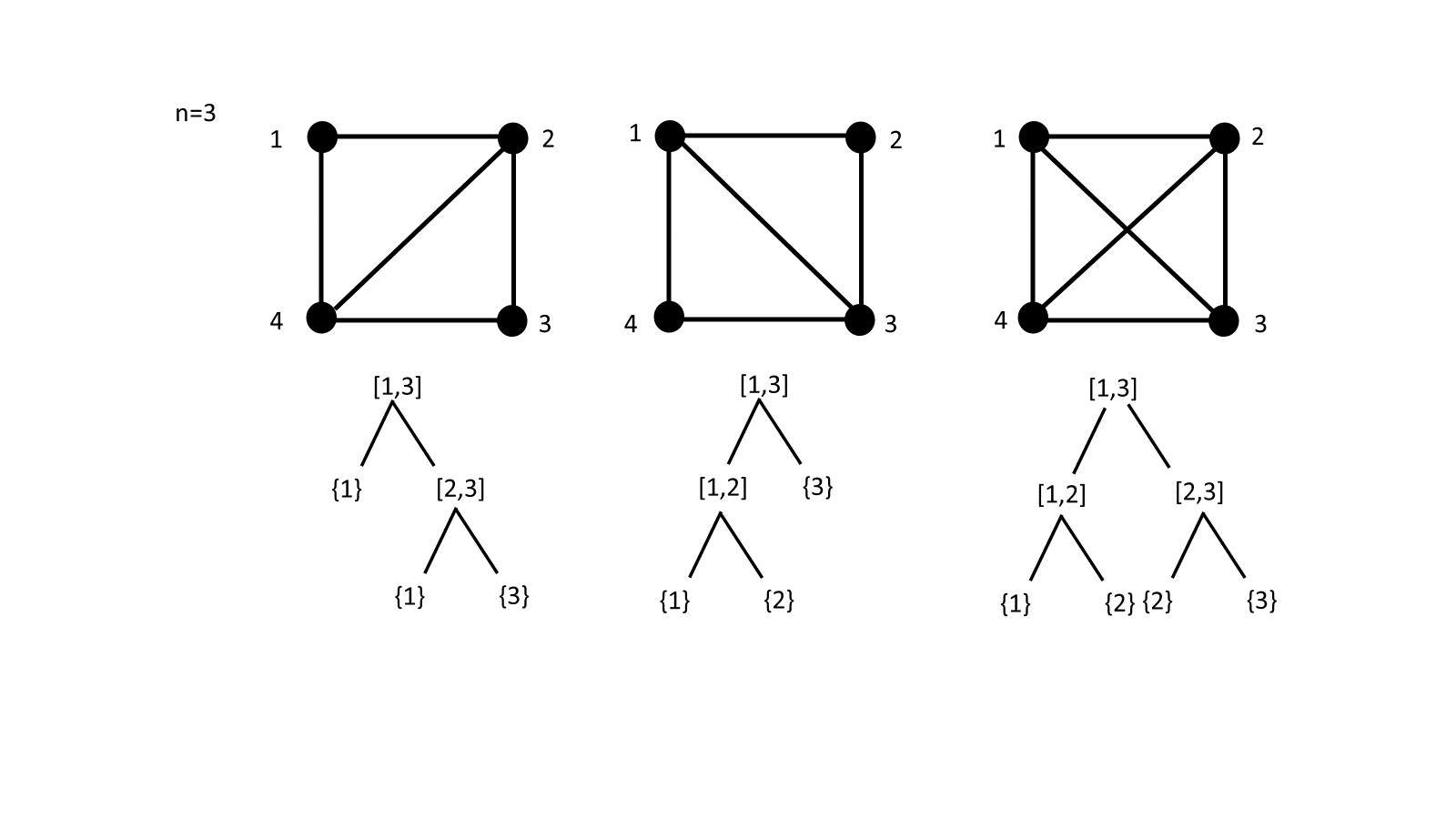}
   
    \caption{The bijection for small values of $n$}
    \label{all}
\end{figure}

\begin{figure}
    \centering
   \includegraphics[scale=0.19]{all1.png}
   \includegraphics[scale=0.19]{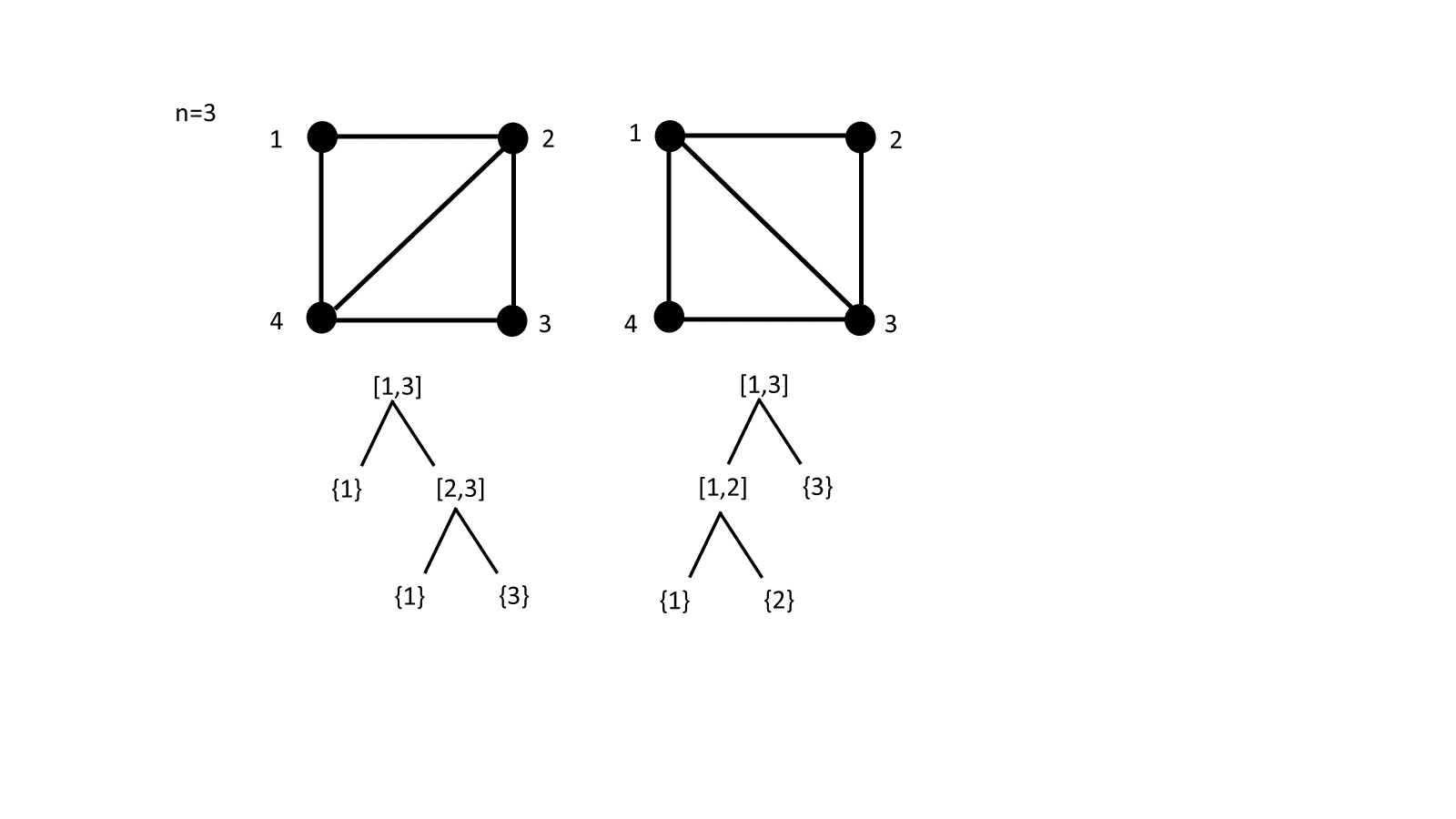}
   \includegraphics[scale=0.19]{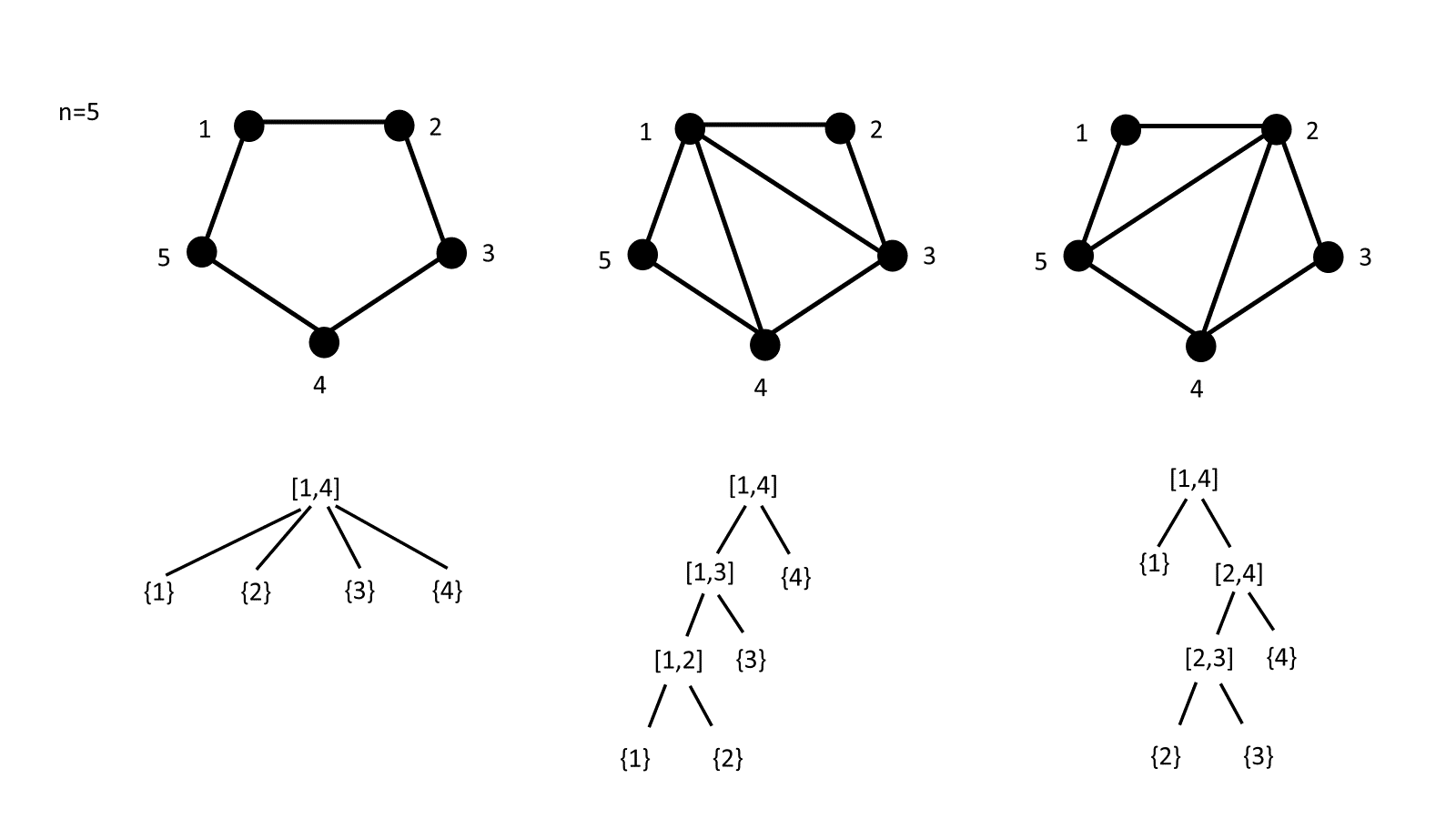}
   \includegraphics[scale=0.19]{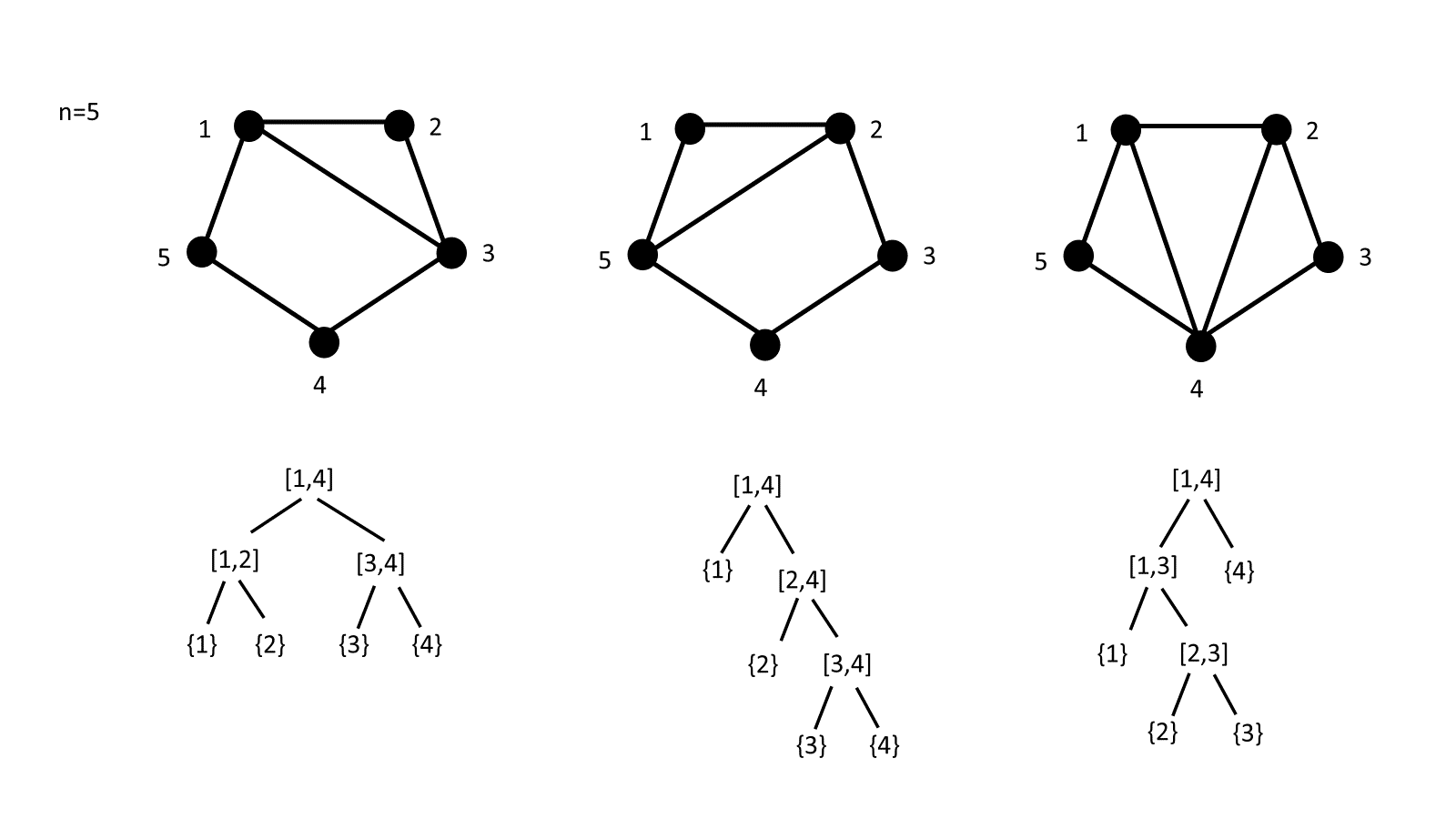}
   \caption{Examples for the bijection of tree intervals for small values of $n$}
    \label{trees}
\end{figure}

\begin{figure}
    \centering
   \includegraphics[scale=0.19]{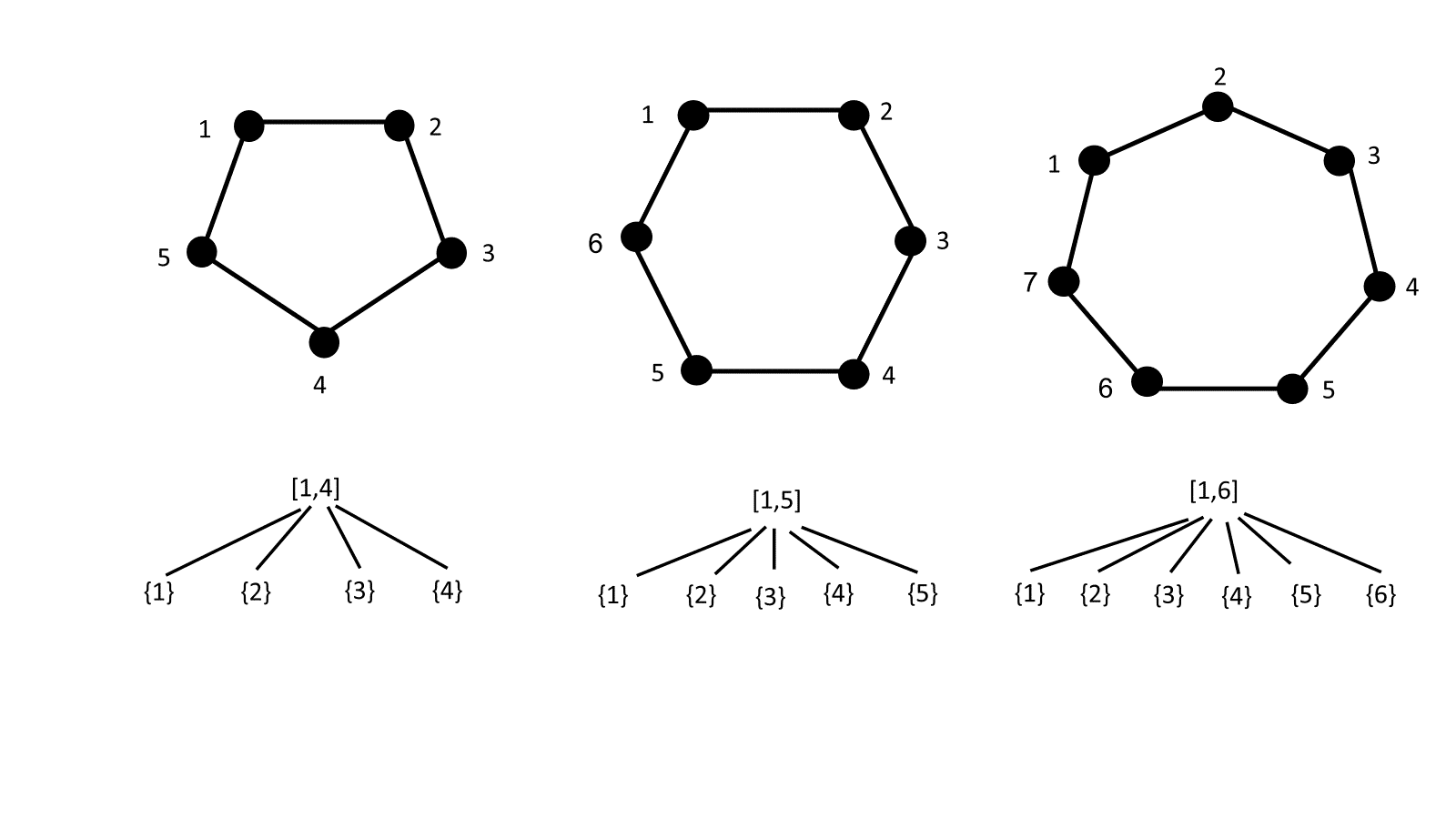}
   \includegraphics[scale=0.19]{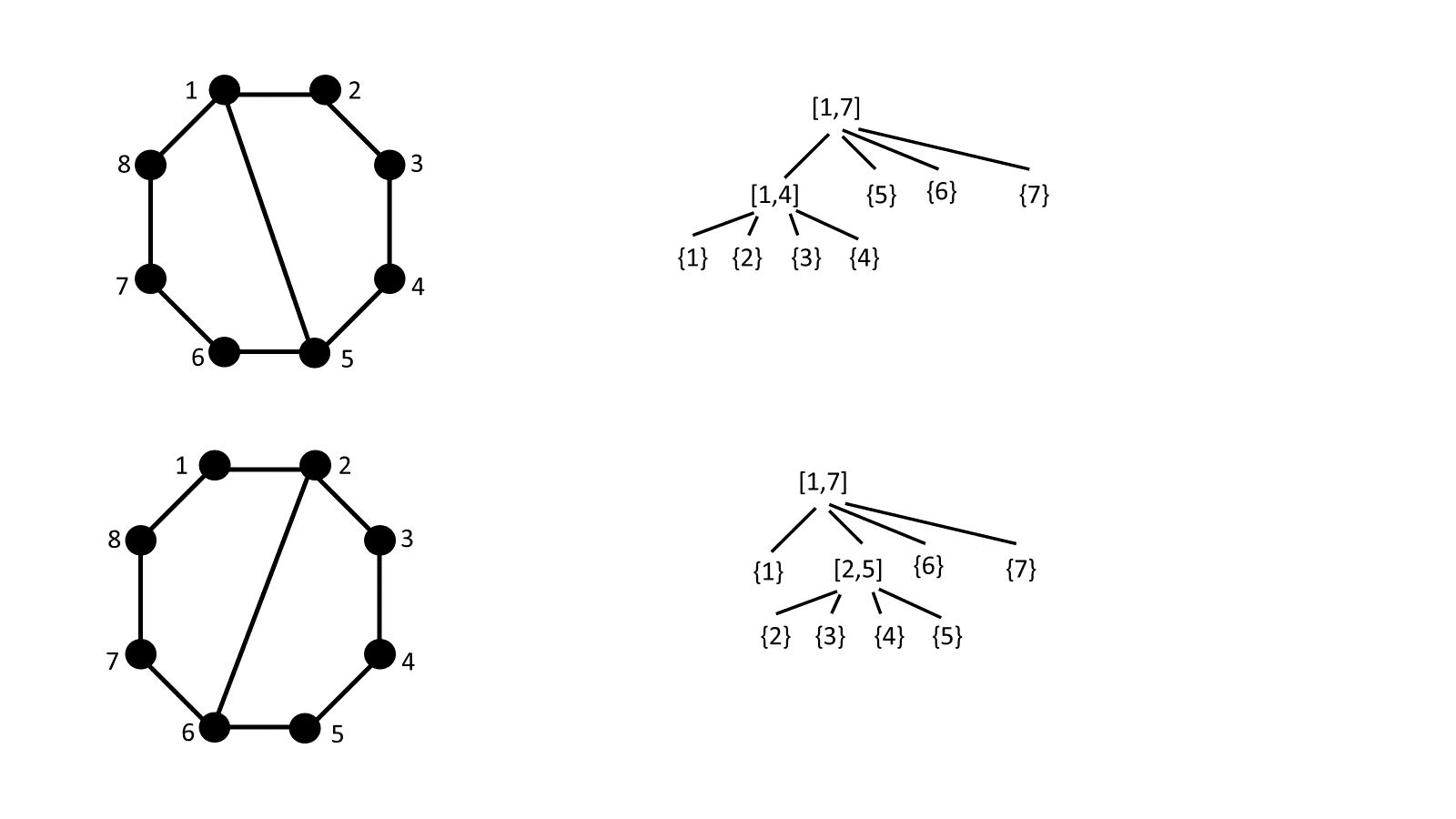}
   \includegraphics[scale=0.19]{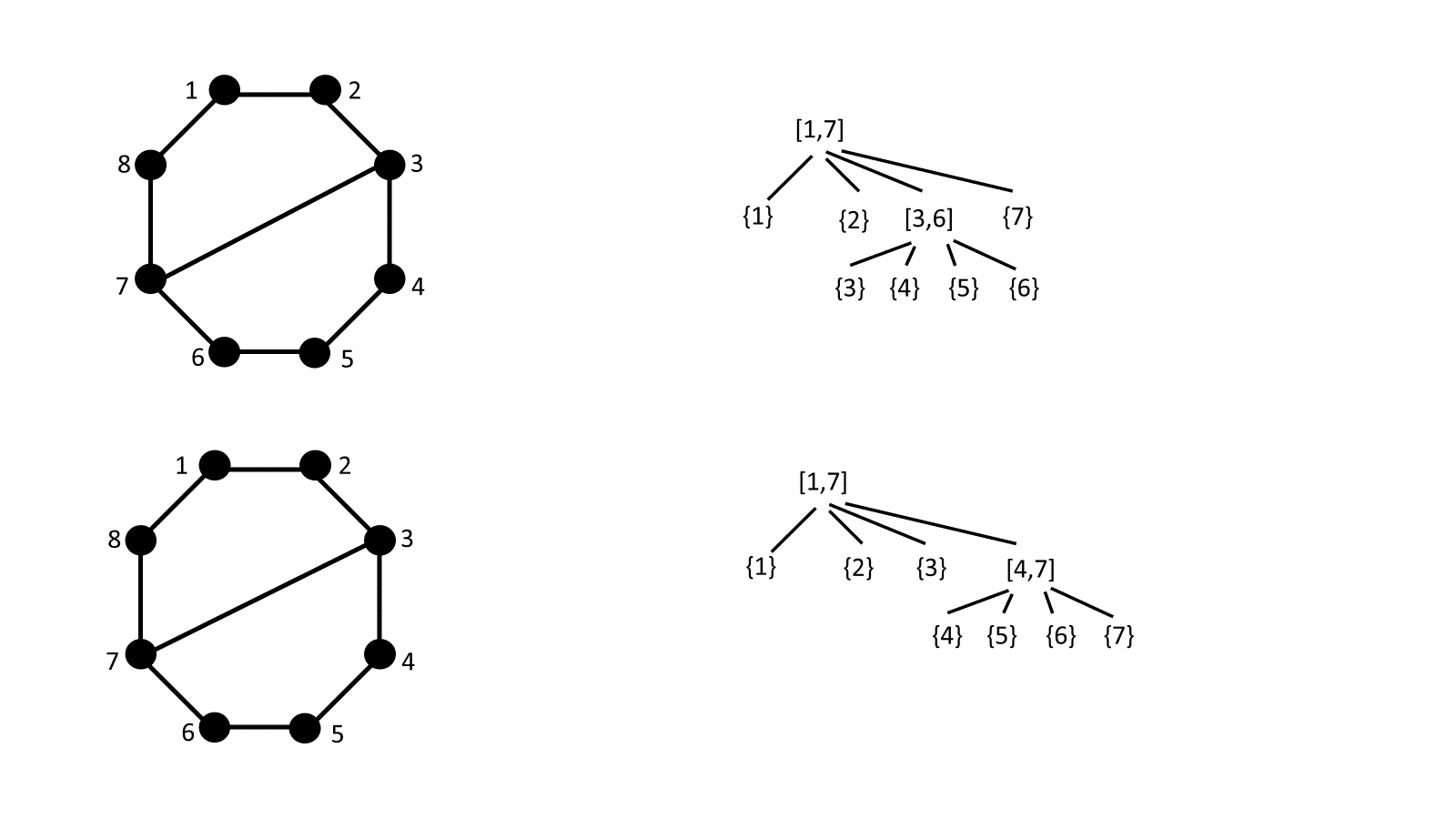}
          \caption{Examples for the bijection of block-wise simple intervals for small values of $n$}
    \label{blockwise}
\end{figure}

\end{document}